%% file: Annuity_LongDeriv.tex
\documentclass[letter,11pt]{elsart}  
\usepackage{amsmath,amsfonts, graphicx, float}
\usepackage[dvips]{epsfig}
\usepackage{subfigure}
\usepackage{caption}
\usepackage{natbib}
\usepackage{eurosym}
\usepackage{natbib}
\usepackage{textcomp}
\usepackage{endnote}
\usepackage{lscape,verbatim}

\usepackage{
rotating,ctable}
\usepackage[latin1]{inputenc}
\usepackage[dvips]{epsfig}
\usepackage{subfigure}
\usepackage{pdflscape}
\usepackage{calc}
\usepackage[section]{placeins}
\numberwithin{equation}{section}

\newcommand{\be}{\begin{equation}}
\newcommand{\ee}{\end{equation}}

\renewcommand{\arraystretch}{0.7}
\newcommand{\bi}{\begin{itemize}}

\newcommand{\ei}{\end{itemize}}

\newcommand{\bc}{\begin{center}}
\newcommand{\ec}{\end{center}}
\newcommand{\ben}{\begin{eqnarray}}
\newcommand{\een}{\end{eqnarray}}
\newcommand{\benn}{\begin{eqnarray*}}
\newcommand{\eenn}{\end{eqnarray*}}

\newcommand{\beq}{\begin{equation}}
\newcommand{\eeq}{\end{equation}}

\newtheorem{proposition}{Proposition}[section]

\newenvironment{proof}{\begin{trivlist}
    \item[\hskip\labelsep{\it Proof.}]}{$\hfill\Box$ \end{trivlist}}
\allowdisplaybreaks

\usepackage{anysize}
\marginsize{18mm}{18mm}{15mm}{10mm} 

\begin{document}

\begin{frontmatter}

\title{Managing Systematic Mortality Risk in\\ Life Annuities: An Application of Longevity Derivatives}

This version: \today

\author[address1]{Man Chung Fung},
\author[address2]{Katja Ignatieva},
\author[address3]{Michael Sherris}

\begin{abstract}
This paper assesses the hedge effectiveness of an index-based longevity swap and a longevity cap. Although swaps are a natural instrument for hedging longevity risk, derivatives with non-linear pay-offs, such as longevity caps, also provide downside protection. A tractable stochastic mortality model with age dependent drift and volatility is developed and analytical formulae for prices of these longevity derivatives are derived. Hedge effectiveness is considered for a hypothetical life annuity portfolio. The hedging of the life annuity portfolio is comprehensively assessed for a range of assumptions for the market price of longevity risk, the term to maturity of the hedging instruments, as well as the size of the underlying annuity portfolio. The model is calibrated using Australian mortality data. The results provide a comprehensive analysis of longevity hedging, highlighting the risk management benefits and costs of linear and nonlinear payoff structures.

\end{abstract}

\address[address1]{School of Risk and Actuarial Studies, University of New South Wales, Sydney, Australia (m.c.fung@unsw.edu.au)}
\address[address2]{School of Risk and Actuarial Studies, University of New South Wales, Sydney, Australia (k.ignatieva@unsw.edu.au)}
\address[address3]{CEPAR, School of Risk and Actuarial Studies, University of New South Wales, Sydney, Australia (m.sherris@unsw.edu.au)}


\maketitle\thispagestyle{empty}

\end{frontmatter}

\newpage


\input{Introduction}
\input{MortalityModel}

\input{LongevityOptions}

\input{NumericalExamples}

\input{Conclusion}
\appendix
\input{AppendixA2}

\bibliographystyle{elsart-harv}
\bibliography{mcf}

\end{document}

%% file: Introduction.tex
\section{Introduction}

Securing a comfortable living after retirement is fundamental to the majority of the working population around the
world. A major risk in retirement, however, is the possibility that
retirement savings will be outlived. Products that provide
guaranteed lifetime income, such as life annuities, need to be offered in a cost effective way while maintaining the long run solvency of the provider. Annuity providers
and pension funds need to manage the systematic
mortality risk\footnote{From an annuity provider's perspective,
longevity risk modelling can lead to a (stochastically) over- or
underestimation of survival probabilities for all annuitants. For
this reason longevity risk is also referred to as the systematic
mortality risk.}, associated with random changes in the underlying
mortality intensity, in a life annuity or pension portfolio. Systematic
mortality risk cannot be diversified away with increasing portfolio
size, while idiosyncratic mortality risk, representing the
randomness of deaths in a portfolio with fixed mortality intensity,
is diversifiable.

Reinsurance has been important in managing longevity risk for annuity and pension providers. However, there are concerns that reinsurers have a limited risk appetite and
are reluctant to take this ``toxic" risk (\cite{BlakeCaDoMa06}). In
fact, even if they were willing to accept the risk, the reinsurance
sector is not deep enough to absorb the vast scale of longevity risk
currently undertaken by annuity providers and pension
funds.\footnote{It is estimated that pension assets for the 13
largest major pension markets have reached nearly 30 trillions in
2012 (Global Pension Assets Study 2013, Towers Watson).} The
sheer size of capital markets and an almost zero correlation between
financial and demographic risks, suggests that they will increasingly take a role in the risk management of longevity risk. 

The first
generation of capital market solutions for longevity risk, in the form of mortality and longevity bonds ( \cite{BlakeBurrows2001},
\cite{Blake2006a} and \cite{BauerBoRu10})\footnote{Of particular
interest is an attempt to issue the EIB longevity bond by the
European Investment Bank (EIB) in 2004, which was underwritten by
BNP Paribas. This bond was not well received by investors and could
not generate enough demand to be launched due to its deficiencies,
as outlined in \cite{Blake2006a}.}, gained limited success. 

The second generation involving forwards and swaps have attracted increasing interest (\cite{Blakeetal13}). Index-based instruments aim to mitigate systematic mortality risk, and have the potential to be less costly and are designed to allow trading as standardised contracts (\cite{Blakeetal13}).
Unlike the bespoke or customized hedging instruments such as reinsurance, they do not cover idiosyncratic
mortality risk and give rise to basis risk
(\cite{LiHa11}). Since idiosyncratic
mortality risk is reduced for larger portfolios, portfolio size is an important factor that
determines the hedge effectiveness of index-based instruments.

Longevity
derivatives with a linear payoff, including q-forwards and S-forwards,
have as an underlying the mortality and the survival rate,
respectively (\cite{LLMA10a}). Their hedge effectiveness has been considered in \cite{NgaiSh11} who study the effectiveness of static hedging of
longevity risk in different annuity portfolios. They consider a range of
longevity-linked instruments including q-forwards, longevity
bonds and longevity swaps as hedging instruments to
mitigate longevity risk and demonstrate their benefits in reducing longevity risk. \cite{LiHa11} also consider
hedging longevity risk with a portfolio of q-forwards. They highlight
basis risk as one of the obstacles in the development of
an index-based longevity market. 

Longevity derivatives with
a nonlinear payoff structure have not received a great deal of attention to date. \cite{BoyerSt13} evaluate European and American type
survivor options using simulations and \cite{WangYa13} propose and price
survivor floors under an extension of the Lee-Carter model. These authors do not consider the hedge effectiveness of longevity options and longevity swaps as hedging instruments.

Although dynamic hedging has been considered, because of the lack of liquid markets in longevity risk, static hedging remains the only realistic option for annuity providers. \cite{Cairns11} considers q-forwards and a discrete-time delta hedging strategy, and compares it
with static hedging. The lack of analytical formulas
for pricing q-forwards and its derivatives, known as ``Greeks",
can be a significant problem in assessing hedge effectiveness since simulations within simulations
are required for dynamic hedging strategies. The importance of tractable models
has also been emphasised in \cite{LucianoReVi12} who also consider dynamic hedging for longevity and interest rate risk. \cite{HariWaMeNi08} apply a generalised
two-factor Lee-Carter model to investigate the impact of longevity
risk on the solvency of pension annuities.

This paper provides pricing analysis of longevity derivatives, as
well as their hedge effectiveness. We consider static hedging. A longevity swap and a cap
are chosen as linear and nonlinear products
to compare and assess index-based capital market products management of longevity risk management. The model used for this analysis is a continuous time model for mortality with age based drift and volatility, allowing tractable analytical formulae for pricing and hedging. The analysis is based on a hypothetical life annuity portfolio subject to longevity risk. The paper considers the hedging of longevity risk using a longevity swap
and a longevity cap, a portfolio
of S-forwards and longevity caplets respectively, based on a range of different underlying assumptions for the market
price of longevity risk, the term to maturity of hedging
instruments, as well as the size of the underlying annuity
portfolio.

The paper is organised as follows. Section 2 specifies the
two-factor Gaussian mortality model, and its parameters are estimated
using Australian males mortality data. Section 3 analyses longevity
derivatives, in particular, a longevity swap and a cap, from a
pricing perspective. Explicit pricing formulas are derived under the
proposed two-factor Gaussian mortality model. Section 4 examines various hedging
features and hedge effectiveness of a longevity swap and a cap on a
hypothetical life annuity portfolio exposed to longevity risk.
Section 5 summarises the results and provides concluding remarks.

%% file: MortalityModel.tex
\section{Mortality Model}\label{sec:model}

Let $(\Omega,\mathcal{F}_t=\mathcal{G}_t \vee
\mathcal{H}_t,\mathbb{P})$ be a filtered probability space where
$\mathbb{P}$ is the real world probability measure. The
subfiltration $\mathcal{G}_t$ contains information about the
dynamics of the mortality intensity while death times of individuals
are captured by $\mathcal{H}_t$. It is assumed that the interest
rate $r$ is constant where $B(0,t) = e^{-r\,t}$ denotes the price of
a $t$-year zero coupon bond, and our focus is on the modelling of
stochastic mortality.

\subsection{Model Specification}

For the purpose of financial risk management applications one
requires stochastic mortality model that is tractable, and is able to
capture well the mortality dynamics for different ages. We work
under the affine mortality intensity framework and assume the
mortality intensity to be Gaussian such that analytical prices can
be derived for longevity options, as described in
Section~\ref{sec:derivativepricing}. Gaussian mortality models have
been considered in \cite{BauerBoRu10} and \cite{BlackburnSh13} within
the forward mortality framework. \cite{LucianoVi08} suggest Gaussian
mortality where the intensity follows the Ornstein-Uhlenbeck
process. In addition, \cite{JevticLuVi13} consider a continuous time
cohort model where the underlying mortality dynamics is Gaussian.

We consider a two-factor Gaussian mortality model for the mortality
intensity process $\mu_{x+t}(t)$ of a cohort aged $x$ at time
$t=0$\footnote{For simplicity of notation we replace $\mu_{x+t}(t)$
by $\mu_x(t)$.}:
\begin{equation}\label{eq:mi}
d\mu_x(t) = dY_1(t)+dY_2(t),
\end{equation}
where
\begin{align}
dY_1(t)&=\alpha_1 Y_1(t)\,dt+\sigma_1\,dW_1(t) \label{eq:dY1}\\
dY_2(t)&=(\alpha \,x+\beta)\,Y_2(t)\,dt+\sigma e^{\gamma x}\,dW_2(t) \label{eq:dY2}
\end{align}
and $dW_1dW_2=\rho\,dt$. The first factor $Y_1(t)$ is a general
trend for the intensity process that is common to all ages. The
second factor $Y_2(t)$ depends on the initial age through the drift
and the volatility terms.\footnote{We can in fact replace $x$ by
$x+t$ in Eq.~\eqref{eq:dY2}. Using $x+t$ will take into account the
empirical observation that the volatility of mortality tends to
increase along with age $x+t$ (Figures ~\ref{fig:mort_rate} and
~\ref{fig:diff_mort_rate}). However, for a Gaussian process the
intensity will have a non-negligible probability of reaching
negative value when the volatility from the second factor ($\sigma
e^{\gamma(x+t)}$) becomes very high, which occurs for example when $x+t > 100$
(given $\gamma>0$). Using $x$ instead of $x+t$ will also make the
result in Section~\ref{sec:derivativepricing} easy to interpret.
For these reasons we assume that the second factor $Y_2(t)$ depends on
the initial age $x$ only.} The initial values $Y_1(0)$ and $Y_2(0)$
of the factors are denoted by $y_1$ and $y^x_2$, respectively. The
model is tractable and for a specific choice of the parameters (when
$\alpha=\gamma =0$) has been applied to short rate modelling in
\cite{BrigoMe07}.

\begin{proposition}\label{prop:sp}
Under the two-factor Gaussian mortality model (Eq.~\eqref{eq:mi} -
~\eqref{eq:dY2}), the $(T-t)$- year expected survival probability of
a person aged $x+t$ at time $t$, conditional on filtration
$\mathcal{F}_t$, is given by
\begin{align}\label{eq:sp}
S_{x+t}(t,T)&\stackrel{\text{def}}{=}E^\mathbb{P}_t\left(e^{-\int^T_t \mu_{x}(v)dv}\right)=e^{\frac{1}{2}\Gamma(t,T) - \Theta(t,T)}, 
\end{align}
where, using $\alpha_2 = \alpha x + \beta$ and $\sigma_2 = \sigma e^{\gamma x}$,
\begin{align}
\Theta(t,T) &= \frac{(e^{\alpha_1 (T-t)}-1)}{\alpha_1}Y_1(t)+\frac{(e^{\alpha_2 (T-t)}-1)}{\alpha_2}Y_2(t) \label{eq:Theta}\hspace{+0.2cm}\mbox{and}\\
\Gamma(t,T) &=\sum^2_{k=1}\frac{\sigma^2_k}{\alpha^2_k}\left(T-t-\frac{2}{\alpha_k}e^{\alpha_k(T-t)}+\frac{1}{2\alpha_k}
e^{2\alpha_k(T-t)}+\frac{3}{2\alpha_k}\right)+ \notag \\
&\frac{2\rho\sigma_1\sigma_2}{\alpha_1\alpha_2}\left(T-t-
\frac{e^{\alpha_1(T-t)}-1}{\alpha_1}-\frac{e^{\alpha_2(T-t)}-1}{\alpha_2}+\frac{e^{(\alpha_1+\alpha_2)(T-t)}-1}{\alpha_1+\alpha_2}\right).\label{eq:Gamma}
\end{align}
are the mean and the variance of the integral $\int^T_t \mu_{x}(v)\,dv$, which is Gaussian distributed, respectively.
\end{proposition}
We will use the fact that the integral $\int^T_t \mu_{x}(v)\,dv$ is
Gaussian with known mean and variance to derive
analytical pricing formulas for longevity options in
Section~\ref{sec:derivativepricing}.
\begin{proof}
Solving Eq.~\eqref{eq:dY1} to obtain an integral form of $Y_1(t)$, we have
\begin{equation}\label{eq:2terms}
\int^T_tY_1(u)\,du=\int^T_tY_1(t)e^{\alpha_1(u-t)}du+\int^T_t\sigma_1\int^u_te^{\alpha_1(u-v)}dW_1(v)du.
\end{equation}
The first term in Eq.~\eqref{eq:2terms} can be simplified to
\begin{equation*}
\int^T_tY_1(t)e^{\alpha_1(u-t)}du=\frac{\left(e^{\alpha_1(T-t)}-1\right)}{\alpha_1}Y_1(t).
\end{equation*}
For the second term, we have
\begin{align*}
&\sigma_1 \int^T_te^{\alpha_1 u} \int^u_t e^{-\alpha_1v}dW_1(v)du=\sigma_1\int^T_t\int^u_t e^{-\alpha_1v}dW_1(v)d_u\left(\frac{1}{\alpha_1}e^{\alpha_1 u}\right) \\
&=\frac{\sigma_1}{\alpha_1}\int^T_td_u\left(e^{\alpha_1 u}\int^u_t e^{-\alpha_1 v} dW_1(v)\right) -\frac{\sigma_1}{\alpha_1}\int^T_te^{\alpha_1 u} d_u\left(\int^u_t e^{-\alpha_1 v} dW_1(v)\right) \\
&= \frac{\sigma_1}{\alpha_1}e^{\alpha_1T}\int^T_te^{-\alpha_1 u }dW_1(u)-\frac{\sigma_1}{\alpha_1}
\int^T_t e^{\alpha_1 u}e^{-\alpha_1u}dW_1(u)=\frac{\sigma_1}{\alpha_1}\int^T_te^{\alpha_1(T-u)}-1\,dW_1(u),
\end{align*}
where stochastic integration by parts is applied in the second
equality. 

To obtain an integral representation for $Y_2(t)$, we follow the same steps as above, replacing
$Y_1(t)$ by $Y_2(t)$ in Eq.~\eqref{eq:2terms}. It is then
straightforward to notice that
\begin{equation}\label{eq:mort2term}
\int^T_t\mu_x(u)\,du = \int^T_t Y_1(u)+Y_2(u)\,du
\end{equation}
is a Gaussian random variable with mean $\Theta(t,T)$
(Eq.~\eqref{eq:Theta}) and variance $\Gamma(t,T)$
(Eq.~\eqref{eq:Gamma}). Equation \eqref{eq:sp} is obtained by
applying the moment generating function of a Gaussian random
variable.
\end{proof}

\subsection{Parameter Estimation}\label{subsec:estimation}

The discretised process, where the intensity is assumed to be
constant over each integer age and calendar year, is approximated by
the central death rates $m(x,t)$ (\cite{WillsSh11}).
Figure~\ref{fig:mort_rate} displays Australian male central death
rates $m(x,t)$ for years $t=1970,1971,\dots,2008$ and ages $x =
60,61,\dots,95$. Figure~\ref{fig:diff_mort_rate} shows the
difference of the central death rates $\Delta m(x,t) =
m(x+1,t+1)-m(x,t)$. The variability of $\Delta m(x,t)$ is evidently
increasing with increasing age $x$, which leads to the anticipation
that $\gamma>0$. Furthermore, for a fixed age $x$, there is a slight
improvement in central death rates for more recent years, compared
to the past.

\begin{figure}[h]
  \begin{center}
   \includegraphics[width=12cm, height=9cm]{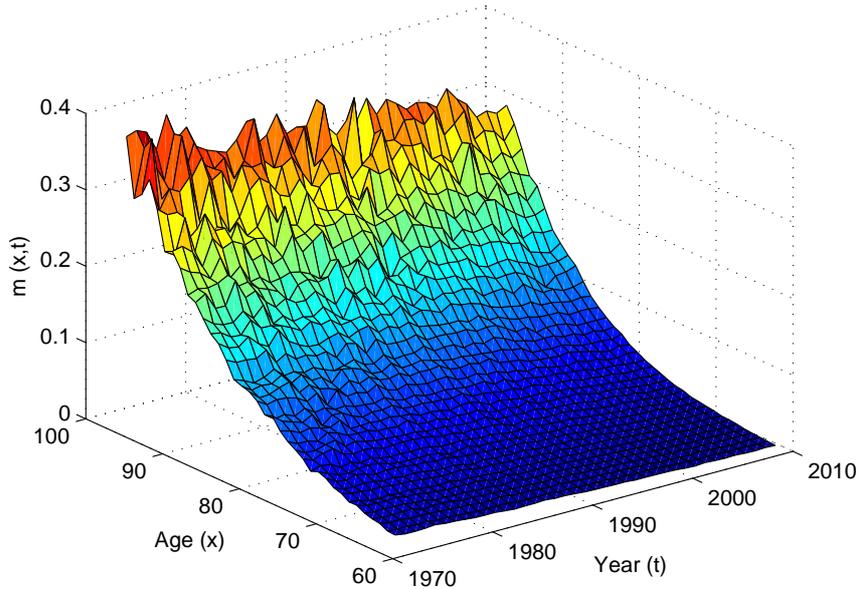}
  \caption{Australian male central death rates $m(x,t)$ where $t =1970,1971,\dots,2008$ and $x = 60,61,\dots,95$.}\label{fig:mort_rate}
    \end{center}
\end{figure}

\begin{figure}[h]
  \begin{center}
   \includegraphics[width=12cm, height=9cm]{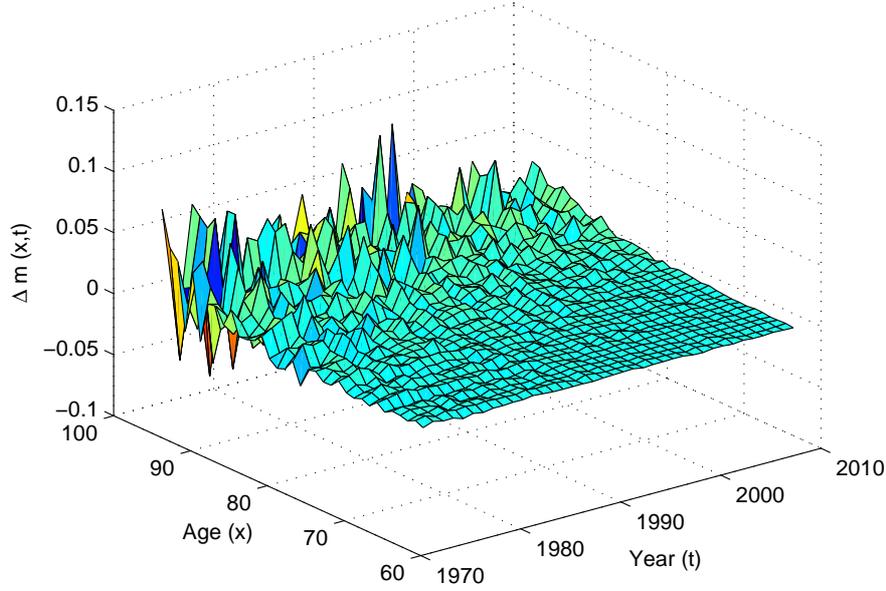}
  \caption{Difference of the central death rates $\Delta m(x,t) = m(x+1,t+1)-m(x,t)$ where $t =1970,1971,\cdots,2007$ and $x = 60,61,\cdots,94$.}\label{fig:diff_mort_rate}
    \end{center}
\end{figure}

The parameters $\{\sigma_1,\sigma, \gamma,\rho\}$, which determine
the volatility of the intensity process, are estimated as described below. As in \cite{JevticLuVi13}, we aim to estimate parameters using the method of least squares, thus, calibrating the model to the mortality surface. However, we take advantage of the fact that a Gaussian model is employed where the variance of the model can be calculated explicitly and thus, we capture the diffusion part of the process by matching the variance of the model to mortality data. Specifically, the implemented procedure is as specified below:

\begin{enumerate}
  \item Using empirical data for ages $x = 60, 65,\dots, 90$ we evaluate the sample variance of $\Delta m(x,t)$ across time, denoted by $\text{Var}(\Delta m_x)$.
  \item The model variance $\text{Var}(\Delta \mu_x)$ for age $x$ is given by
      \begin{align}\nonumber
      \text{Var}(\Delta \mu_x) &= \text{Var}(\sigma_1\Delta W_1 + \sigma e^{\gamma x} \Delta W_2)\\
      &=\left(\sigma_1^2+2\sigma_1\sigma\rho e^{\gamma x}+\sigma^2e^{2\gamma x}\right)\Delta t.
      \end{align}
      Since the difference between the death rates is computed in yearly terms, we set $\Delta t =1$.
  \item The parameters $\{\sigma_1,\sigma, \gamma,\rho\}$ are then estimated by fitting the model variance $\text{Var}(\Delta \mu_x)$ to the sample variance $\text{Var}(\Delta m_x)$ for ages $x = 60, 65,\dots, 90$ using least squares estimation, that is, by minimising
      \begin{equation}
      \sum^{90}_{x=60,65\dots}\left(\text{Var}(\Delta \mu_x|\sigma_1,\sigma,\gamma,\rho)-\text{Var}(\Delta m_x)\right)^2
      \end{equation}
      with respect to the parameters $\{\sigma_1,\sigma, \gamma,\rho\}$.
\end{enumerate}
The remaining parameters $\{\alpha_1,\alpha, \beta, y_1, y^{65}_2,
y^{75}_2\}$ are then estimated as described below\footnote{We
calibrate the model for ages 65 and 75 simultaneously to obtain
reasonable values for $\alpha$ and $\beta$ since the drift of the
second factor $Y_2(t)$ is age-dependent.}:
\begin{enumerate}
  \item From the central death rates, we obtain empirical survival curves for cohorts aged 65 and 75 in 2008. The survival curve is obtained by setting
      \begin{equation}
      \hat{S}_x(0,T)=\prod^T_{v=1}(1-m(x+v-1,0))
      \end{equation}
      where $m(x,t)$ is the central death rate of an $x$ years old at time $t$.\footnote{Here $t=0$ represents calendar year 2008 and we approximate the 1-year survival probability $e^{-m(x+v-1,0)}$ by $1-m(x+v-1,0)$.}
  \item The parameters $\{\alpha_1,\alpha, \beta, y_1, y^{65}_2, y^{75}_2\}$ are then estimated by fitting the survival curves ($S_x(0,T)$) of the model to the empirical survival curves using least squares estimation, that is, by minimising
      \begin{equation}
      \sum_{x=65,75}\sum^{T_x}_{j=1}\left(\hat{S}_x(0,j)-S_x(0,j)\right)^2
      \end{equation}
      where $T_{65}=31$ and $T_{75}=21$, with respect to the parameters $\{\alpha_1,\alpha, \beta, y_1, y^{65}_2, y^{75}_2\}$.
\end{enumerate}

The estimated parameters are reported in Table~\ref{table:modelparas}. 
Since $\gamma>0$ we observe that the volatility of the process is higher for older (initial) age $x$.

\begin{table}[!ht]
\center \setlength{\tabcolsep}{1em}
\renewcommand{\arraystretch}{1.1}
\small{\caption{Estimated model parameters.}\label{table:modelparas}
\begin{tabular}{lllll}
\hline
\hline
$\sigma_1$  & $\sigma$ & $\gamma$ & $\rho$ & $\alpha_1$ \\
\hline
0.0022465 & 0.0000002  & 0.129832  & -0.795875 & 0.0017508   \\
\hline
\hline
 $\alpha$ & $\beta$ & $y_1$ & $y^{65}_2$ & $y^{75}_2$ \\
\hline
0.0000615  & 0.120931 & 0.0021277 & 0.0084923 & 0.0294695 \\
\hline
\hline
\end{tabular}}
\end{table}

\begin{figure}[h]
  \begin{center}
   \includegraphics[width=16cm, height=12cm]{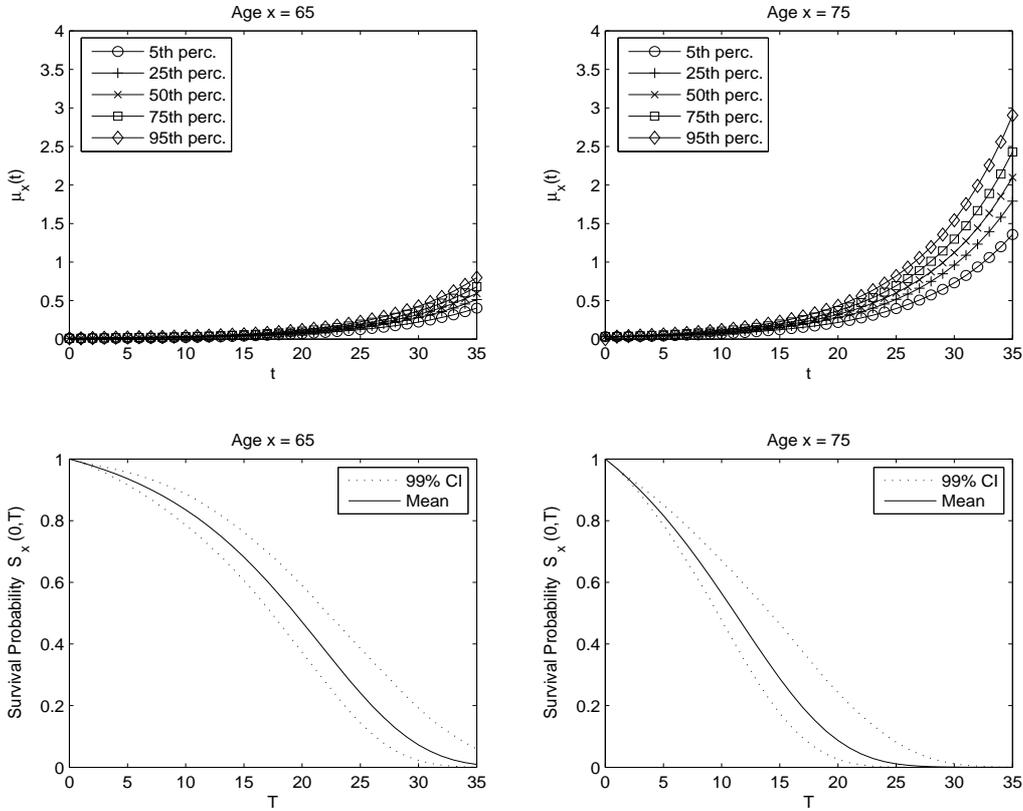}
  \caption{Percentiles of the simulated intensity processes $\mu_{65}(t)$ and $\mu_{75}(t)$ for Australian males aged $65$ (upper left panel) and $75$ (upper right panel) in 2008, with their corresponding survival probabilities (the mean and the $99\%$ confidence bands) for a $65$ years old (lower left panel) and $75$ years old (lower right panel).}\label{fig:sim_mu}
    \end{center}
\end{figure}

The upper panel of Figure~\ref{fig:sim_mu} shows the percentiles of
the simulated mortality intensity for ages 65 and 75 in the left and
the right panel, respectively. One observes that the volatility of the
mortality intensity is higher for a 75 year old compared to a 65
year old. Corresponding survival probabilities are displayed in the
lower panel of Figure~\ref{fig:sim_mu}, together with the $99\%$
confidence bands computed pointwise. As it is pronounced from the
figures, the two-factor Gaussian model specified above, despite its
simplicity, produces reasonable mortality dynamics for ages 65 and
75.

%% file: LongevityOptions.tex
\section{Analytical Pricing of Longevity Derivatives}\label{sec:derivativepricing}

We consider longevity derivatives with different payoff structures
including longevity swaps, longevity caps and longevity floors.
Closed form expressions for prices of these longevity derivatives
are derived under the assumption of the two-factor Gaussian
mortality model introduced in Section~\ref{sec:model}. These
instruments are written on survival probabilities and their
properties are analysed from a pricing perspective.

\subsection{Risk-Adjusted Measure}

For the purpose of no-arbitrage valuation, we require the dynamics of the factors $Y_1(t)$ and $Y_2(t)$ to be written under a risk-adjusted measure.\footnote{Since the longevity market is still in its development stage and hence, incomplete, we assume a risk-adjusted measure exists but is not unique.} To preserve the tractability of the model, we assume that the processes $\tilde{W}_1(t)$ and $\tilde{W}_2(t)$ with dynamics
\begin{align}
d\tilde{W}_1(t)&=dW_1(t) \\
d\tilde{W}_2(t)&=\lambda \sigma e^{\gamma x}\, Y_2(t)\,dt+dW_2(t) \label{eq:QW2}
\end{align}
are standard Brownian motions under a risk-adjusted measure
$\mathbb{Q}$. In Eq.~\eqref{eq:QW2} $\lambda$ represents the market
price of longevity risk.\footnote{For simplicity, we assume that
there is no risk adjustment for the first factor $Y_1$ and $\lambda$
is age-independent.} Under $\mathbb{Q}$ we can write the factor
dynamics as follows:
\begin{align}
dY_1(t)&=\alpha_1 Y_1(t)\,dt+\sigma_1\,d\tilde{W}_1(t) \label{eq:dY1_Q}\\
dY_2(t)&=(\alpha x +\beta-\lambda \sigma e^{\gamma x})\, Y_2(t)\,dt+\sigma_2\,d\tilde{W}_2(t).
\end{align}
The corresponding risk-adjusted survival probability is given by
\begin{equation}\label{eq:ra_sp}
\tilde{S}_{x+t}(t,T) \stackrel{\text{def}}{=}
E^\mathbb{Q}_t\left(e^{-\int^T_t
\mu_x(v)\,dv}\right)=e^{\frac{1}{2}\tilde{\Gamma}(t,T)-\tilde{\Theta}(t,T)}
\end{equation}
where $\alpha_2 = \alpha x+\beta$ is replaced by $(\alpha x + \beta
-\lambda \sigma e^{\gamma x})$ in the expressions for
$\tilde{\Theta}(t,T)$ and $\tilde{\Gamma}(t,T)$, see
Eq.~\eqref{eq:Theta} and Eq.~\eqref{eq:Gamma}, respectively.

Since a liquid longevity market is yet to be developed, we aim to
determine a reasonable value for $\lambda$ based on the longevity
bond announced by BNP Paribas and European Investment Bank (EIB) in
2004 as proposed in  \cite{CairnsBlDo06b} and applied in
\cite{MeyrickeSh14}, see also \cite{WillsSh11}. The BNP/EIB longevity bond is a 25-year bond
with coupon payments linked to a survivor index based on the
realised mortality rates.\footnote{The issue price was determined by
BNP Paribas using anticipated cash flows based on the 2002-based
mortality projections provided by the UK Government Actuary's
Department.} The price of the longevity bond is given by
\begin{equation}
V(0) = \sum^{25}_{T=1} B(0,T)\,e^{\delta\, T} E^\mathbb{P}_0\left({e^{-\int^T_0 \mu_x(v)\,dv}}\right)
\end{equation}
where $\delta$ is a spread, or an average risk premium per
annum\footnote{The spread $\delta$ depends on the term of the bond
and the initial age of the cohort being tracked
(\cite{CairnsBlDo06b}), and $\delta$ is related to but distinct from
$\lambda$, the market price of longevity risk.}, and the T-year
projected survival rate is assumed to be the T-year survival
probability for the Australian males cohort aged 65 as modelled in
Section~\ref{sec:model}, see Eq.~\eqref{eq:sp}. Since the
BNP/EIB bond is priced based on a yield of 20 basis points below
standard EIB rates (\cite{CairnsBlDo06b}), we have the spread
of $\delta = 0.002$.\footnote{The reference cohort for the BNP/EIB
longevity bond is the England and Wales males aged 65 in 2003. Since the longevity derivatives market
is under-developed in Australia, we assume that the same spread of $\delta=0.002$ (as in the UK) is applicable to the
Australian males cohort aged 65 in 2008. Note however that sensitivity analyses will be performed in Section \ref{sec:hypotheticalexample}.}

Under a risk-adjusted measure $\mathbb{Q}(\lambda)$, the price of the longevity bond corresponds to
\begin{equation}
V^{\mathbb{Q(\lambda)}}(0) = \sum^{25}_{T=1} B(0,T)\,E^{\mathbb{Q}(\lambda)}_0\left({e^{-\int^T_0 \mu_x(v)\,dv}}\right).
\end{equation}
Fixing the interest rate to $r=4\%$, we find a model-dependent
$\lambda$, such that the risk-adjusted bond price
$V^{\mathbb{Q(\lambda)}}(0)$ matches the market bond price $V(0)$ as
close as possible. For example, for $\lambda=8.5$ we have $V(0) =
11.9045$ and $V^{\mathbb{Q(\lambda)}}(0)=11.9068$. For more details
on the above procedure refer to \cite{MeyrickeSh14}. In the
following we assume that the risk-adjusted measure $\mathbb{Q}$ is
determined by a unique value of $\lambda$.

\begin{figure}[h]
\begin{center}
\includegraphics[width=12cm, height=8cm]{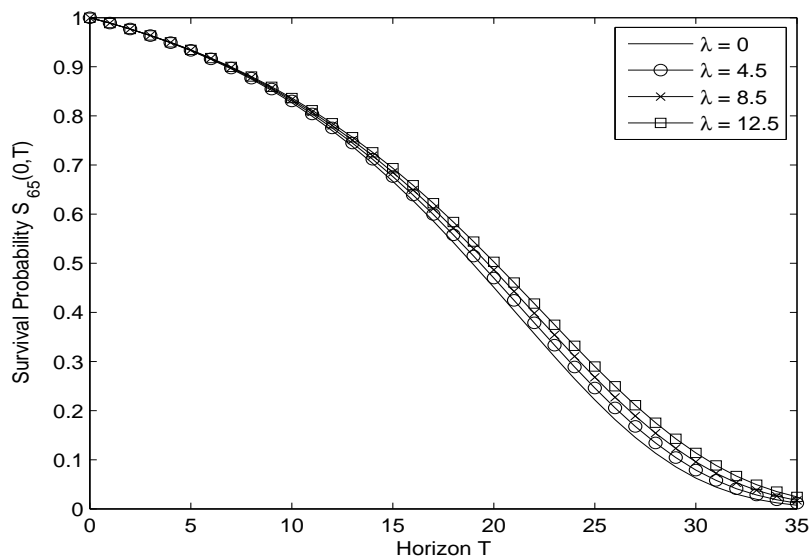}
\caption{Risk-adjusted survival probability with respect to different market price of longevity risk $\lambda$.}\label{fig:RiskAdjusted}
\end{center}
\end{figure}

Figure~\ref{fig:RiskAdjusted} shows the risk-adjusted survival
probabilities for Australian males aged 65 with respect to different
values of the market price of longevity risk $\lambda$. As one observes from the figure, a larger (positive) value of $\lambda$ leads to an improvement in survival probability, while a smaller values of $\lambda$ indicate a decline in survival probability under the risk-adjusted measure $\mathbb{Q}$.

\subsection{Longevity Swaps}\label{subsec:swaps}

A longevity swap involves counterparties swapping fixed payments for
payments linked to the number of survivors in a reference population
in a given time period, and can be thought of as a portfolio of
S-forwards, see \cite{Dowd2003}. An S-forward, or `survivor' forward has been developed by \cite{LLMA10}. Longevity swaps can be regarded as a stream of S-forwards with different maturity dates. One of the advantages of using S-forwards is that there is no initial capital requirement at the inception of the contract and cash flows occur only at maturity. 

Consider an annuity provider who has an obligation to
pay an amount dependent on the number of survivors, and hence, survival probability of a cohort at time $T$. If
longevity risk is present, the survival probability is
stochastic. In order to protect himself from a larger-than-expected
survival probability, the provider can enter into an S-forward
contract paying a fixed amount $K\in(0,1)$ and receiving an amount
equal to the realised survival probability
$\exp{\{-\int^{T}_0\mu_x(v)\,dv\}}$ at time $T$. In doing so, the survival probability that the provider is exposed
to is certain, and corresponds to some fixed value $K$. If the contract is priced in such a way that
there is no upfront cost at the inception, it must
hold that
\begin{align}
B(0,T)\,E^\mathbb{Q}_0\left(e^{-\int^{T}_0\mu_x(v)\,dv}-K(T)\right)=0
\end{align}
under the risk-adjusted measure $\mathbb{Q}$. Thus, the fixed amount can be identified to be the risk-adjusted survival probability, that is,
\begin{align}\label{eq:KT}
K(T)=E^\mathbb{Q}_0\left(e^{-\int^{T}_0\mu_x(v)\,dv}\right).
\end{align}
Assuming that there is a positive market price of longevity risk,
the longevity risk hedger who pays the fixed leg and receives the
floating leg bears the cost for entering an S-forward.\footnote{The
risk-adjusted survival probability will be larger than the ``best
estimate" $\mathbb{P}$-survival probability if a positive market
price of longevity risk is demanded, see Figure~\ref{fig:RiskAdjusted}.}
Following terminology in \cite{BiffisBlPiAr14}, the amount
$K(T)=\tilde{S}_x(0,T)$ can be referred to as the swap rate of an
S-forward with maturity $T$. In general, the mark-to-market price
process $F(t)$ of an S-forward with fixed leg $K$ (not necessarily
$K(T)$ as in Eq.~\eqref{eq:KT}) is given by
\begin{align}
F(t) &= B(t,T) E^\mathbb{Q}_t\left(e^{-\int^{T}_0\mu_x(v)\,dv}-K\right)\nonumber\\
&=B(t,T)E^\mathbb{Q}_t\left(e^{-\int^{t}_0\mu_x(v)\,dv}e^{-\int^{T}_t\mu_x(v)\,dv}-K\right)\nonumber\\
&= B(t,T)\left(\bar{S}_x(0,t)\,\tilde{S}_{x+t}(t,T)-K\right) \label{eq:F_t}
\end{align}
for $t\in [0,T]$. The quantity
\begin{equation}\label{eq:survivorindex}
\bar{S}_x(0,t)=e^{-\int^{t}_0\mu_x(v)\,dv}\vert_{\mathcal{F}_t}
\end{equation}
is the realised survival probability, or the survivor index for the
cohort, which is observable given $\mathcal{F}_t$.

The term $\bar{S}_x(0,t)\,\tilde{S}_{x+t}(t,T)$ that appears in
Eq.~\eqref{eq:F_t} has a natural interpretation. Given information
$\mathcal{F}_0$ at time $t=0$, this term becomes
$\tilde{S}_{x}(0,T)$, which is the risk-adjusted survival
probability. As time moves on and more information $\mathcal{F}_t$,
with $t\in(0,T)$, is revealed, the term
$\bar{S}_x(0,t)\,\tilde{S}_{x+t}(t,T)$ is a product of the realised
survival probability of the first $t$ years, and the risk-adjusted
survival probability in the next $(T-t)$ years. At maturity $T$,
this product becomes the realised survival probability up to time
$T$. In order words, one can think of
$\bar{S}_x(0,t)\,\tilde{S}_{x+t}(t,T)$ as the $T$-year risk-adjusted
survival probability with information known up to time $t$.

The price process $F(t)$ in Eq.~\eqref{eq:F_t} depends on the swap
rate $\tilde{S}_{x+t}(t,T)$ of an S-forward written on the same
cohort that is now aged $(x+t)$ at time $t$, with time to maturity
$(T-t)$. If a liquid longevity market was developed, the swap rate
$\tilde{S}_{x+t}(t,T)$ could be obtained from market data. As
$\bar{S}_x(0,t)$ is observable at time $t$, the mark-to-market price
process of an S-forward could be considered model-independent.
However, since a longevity market is still in its development
stage, market swap rates are not available and a model-based
risk-adjusted survival probability $\tilde{S}_{x+t}(t,T)$ has to be
used instead. An analytical formula for the mark-to-market price of
an S-forward can be obtained if the risk-adjusted survival probability
is expressed in a closed-form, which can be performed, for example, under the two-factor
Gaussian mortality model.

Since a longevity swap is constructed as a portfolio of S-forwards,
the price of a longevity swap is simply the sum of the individual
S-forward prices.

\subsection{Longevity Caps}\label{subsec:caps}

A longevity cap, which is a portfolio of longevity caplets, provides
a similar hedge to a longevity swap but is an option-type instrument.
Consider again a scenario described in Section~\ref{subsec:swaps}
where an annuity provider aims to hedge against
larger-than-expected $T$-year survival probability of a particular
cohort. Alternatively to hedging with an S-forward, the provider can
enter into a long position of a longevity caplet with payoff at time
$T$ corresponding to
\begin{equation}\label{eq:cappayoff}
\max{\left\{\left(e^{-\int^{T}_0\mu_x(v)\,dv}-K\right),0\right\}}
\end{equation}
where $K\in(0,1)$ is the strike price.\footnote{The payoff of a
longevity caplet is similar to the payoff of the option embedded in
the principal-at-risk bond described in \cite{BiffisBl14}.} If the
realised survival probability is larger than $K$, the hedger
receives an amount
$\left(\exp{\{-\int^{T}_0\mu_x(v)\,dv\}}-K\right)$ from the
longevity caplet. This payment can be regarded as a compensation for
the increased payments that the provider has to make in the annuity
portfolio, due to the larger-than-expected survival probability.
There is no cash outflow if the realised survival probability is
smaller than or equal to $K$. In other words, the longevity caplet
allows the provider to ``cap" its longevity exposure at $K$ with no
downside risk. Since a longevity caplet has a non-negative payoff,
it comes at a cost. The price of a longevity caplet
 \begin{equation}\label{eq:caplet_expectedvalue}
C\ell(t; T,K)=B(t,T)E^\mathbb{Q}_t\left(\left(e^{-\int^{T}_0\mu_x(v)\,dv}-K\right)^+\right)
\end{equation}
under the two-factor Gaussian mortality model is obtained in the following Proposition.
\begin{proposition}\label{prop:caplet}
Under the two-factor Gaussian mortality model (Eq.~\eqref{eq:mi}-Eq.~\eqref{eq:dY2}) the price at time $t$ of a longevity caplet $C\ell(t;T,K)$, with maturity $T$ and strike $K$, is given by
\begin{equation}\label{eq:Cl_t}
C\ell(t; T,K)=\bar{S}_t\,\tilde{S}_t\,B(t,T)\,\Phi\left(\sqrt{\tilde{\Gamma}(t,T)}-d\right)-
KB(t,T)\Phi\left(-d\right)
\end{equation}
where $\bar{S}_t=\bar{S}_{x}(0,t)$ is the realised survival probability observable at time $t$, $\tilde{S}_t=\tilde{S}_{x+t}(t,T)$ is the risk-adjusted survival probability in the next $(T-t)$ years, $d=\frac{1}{\sqrt{\tilde{\Gamma}(t,T)}}\left(\ln{\{K/(\bar{S}_t\tilde{S}_t)\}}+\frac{1}{2}\tilde{\Gamma}(t,T)\right)$ and $\Phi(\cdot)$ denotes the cumulative distribution function of a standard Gaussian random variable.
\end{proposition}
\begin{proof}
Under the risk-adjusted measure $\mathbb{Q}$, we have, from Proposition~\eqref{prop:sp}, that
\begin{equation}
L \stackrel{\text{def}}{=} -\int^{T}_t\mu_x(v)dv \sim N(-\tilde{\Theta}(t,T),\tilde{\Gamma}(t,T)).
\end{equation}
Using the simplified notation $\tilde{\Theta}=\tilde{\Theta}(t,T)$, $\tilde{\Gamma}=\tilde{\Gamma}(t,T)$ we can write
\begin{align*}
C\ell(t; T,K) &= B(t,T)E^\mathbb{Q}_t\left((\bar{S}_t\,e^L-K)^+\right) \\
&= B(t,T)\int^\infty_{-\infty} \frac{1}{\sqrt{2\pi\tilde{\Gamma}}}e^{-\frac{1}{2}\left(\frac{\ell+\tilde{\Theta}}{\sqrt{\tilde{\Gamma}}}\right)^2}\left(\bar{S}_t\,e^\ell-K\right)^+\,d\ell \\
&= B(t,T)\int^\infty_{\frac{\ln{K/\bar{S}_t}+\tilde{\Theta}}{\sqrt{\tilde{\Gamma}}}}\frac{1}{\sqrt{2\pi}}e^{-\frac{1}{2}\ell^2}
\left(\bar{S}_t\,e^{\ell\sqrt{\tilde{\Gamma}} -\tilde{\Theta}}-K\right)\,d\ell \\
&= B(t,T)\left(\bar{S}_t\,e^{\frac{1}{2}\tilde{\Gamma}-\tilde{\Theta}}\int^\infty_{\frac{\ln{K/\bar{S}_t}+\tilde{\Theta}}
{\sqrt{\tilde{\Gamma}}}} \frac{1}{\sqrt{2\pi}}e^{-\frac{1}{2}\left(\ell-\sqrt{\tilde{\Gamma}}\right)^2}\,d\ell-
K\int^\infty_{\frac{\ln{K/\bar{S}_t}+\Theta}{\sqrt{\tilde{\Gamma}}}} \frac{1}{\sqrt{2\pi}}e^{-\frac{1}{2}\ell^2}\,d\ell\right).
\end{align*}
Equation~\eqref{eq:Cl_t} follows using properties of $\Phi(\cdot)$ and noticing that $\tilde{S}_t=e^{\frac{1}{2}\tilde{\Gamma} - \tilde{\Theta}}$, that is, $\tilde{\Theta}=\frac{1}{2}\tilde{\Gamma}-\ln{\tilde{S}_t}$.
\end{proof}

Similar to an S-forward, the price of a longevity caplet depends on
the product term $\bar{S}_x(0,t)\,\tilde{S}_{x+t}(t,T)$. In
particular, a longevity caplet is said to be out-of-the-money if
$K>\bar{S}_x(0,t)\,\tilde{S}_{x+t}(t,T)$; at-the-money if
$K=\bar{S}_x(0,t)\,\tilde{S}_{x+t}(t,T)$; and in-the-money if
$K<\bar{S}_x(0,t)\,\tilde{S}_{x+t}(t,T)$. Eq.~\eqref{eq:Cl_t}, is verified using Monte Carlo simulation
summarised in Table~\ref{table:compare_capletprice}, where we set
$r=4\%$, $\lambda=8.5$ and $t=0$. Other parameters are as specified
in Table~\ref{table:modelparas}.

\begin{table}[!ht]
\center \setlength{\tabcolsep}{1em}
\renewcommand{\arraystretch}{1.1}
\small{\caption{Pricing longevity caplet $C\ell (0;T,K)$ by the
formula (Eq.~\eqref{eq:Cl_t}) and by Monte Carlo simulation of
Eq.~\eqref{eq:caplet_expectedvalue}; [ , ] denotes the 95\%
confidence interval.}\label{table:compare_capletprice}
\begin{tabular}{l|ll}
\hline
\hline
(T, K)  & Exact  & { }{ }{ }{ }{ }{ }M.C. Simulation  \\
\hline
(10, 0.6) & 0.15632   & 0.15644 [0.15631, 0.15656]      \\
(10, 0.7) & 0.08929   & 0.08941 [0.08928, 0.08954]      \\
(10, 0.8) & 0.02261   & 0.02262 [0.02250, 0.02275]      \\
(20, 0.3) & 0.08373   & 0.08388 [0.08371, 0.08406]      \\
(20, 0.4) & 0.03890   & 0.03897 [0.03879, 0.03914]      \\
(20, 0.5) & 0.00525   & 0.00530 [0.00522, 0.00539]      \\
\hline
\hline
\end{tabular}}
\end{table}

Following the result of Proposition~\ref{prop:caplet}, the
two-factor Gaussian mortality model leads to the price of a
longevity caplet that is a function of the following variables:
\begin{itemize}
  \item realised survival probability $\bar{S}_x(0,t)$ of the first $t$ years;
  \item risk-adjusted survival probability $\tilde{S}_{x+t}(t,T)$ in the next $T-t$ years;
  \item interest rate $r$;
  \item strike price $K$;
  \item time to maturity $(T-t)$; and
  \item standard deviation $\sqrt{\tilde{\Gamma}(t,T)}$, which is a function of the time to maturity and the model parameters.
\end{itemize}

Since the quantity $\exp\left\{-\int^{T}_0\mu_x(v)\,dv\right\}$ is
log-normally distributed under the two-factor Gaussian mortality
model, Eq.~\eqref{eq:Cl_t} resembles the Black-Scholes formula for
option pricing where the underlying stock price follows a geometric
Brownian motion. In our setup, the stock price at time $t$ is
replaced by the $T$-year risk-adjusted survival probability
$\bar{S}_x(0,t)\,S_{x+t}(t,T)$ with information available up to time
$t$. While the stock is traded and can be modelled directly using
market data, the underlying of a longevity caplet is the survival
probability which is not tradable but can be determined as an output from the
dynamics of mortality intensity. As a result, the role of the stock
price volatility in the Black-Scholes formula is played by the
standard deviation of the integral of the mortality intensity
$\int^T_t \mu_x(v)\,dv$. Since the integral $\int^T_t \mu_x(v)\,dv$
captures the whole history of the mortality intensity $\mu_x(t)$
from $t$ to $T$ under $\mathbb{Q}$, one can interpret the standard
deviation $\sqrt{\tilde{\Gamma}(t,T)}$ as the volatility of the
risk-adjusted aggregated longevity risk of a cohort aged $x+t$ at
time $t$, for the period from $t$ to $T$.

\begin{figure}[h]
  \begin{center}
   \includegraphics[width=8cm, height=6cm]{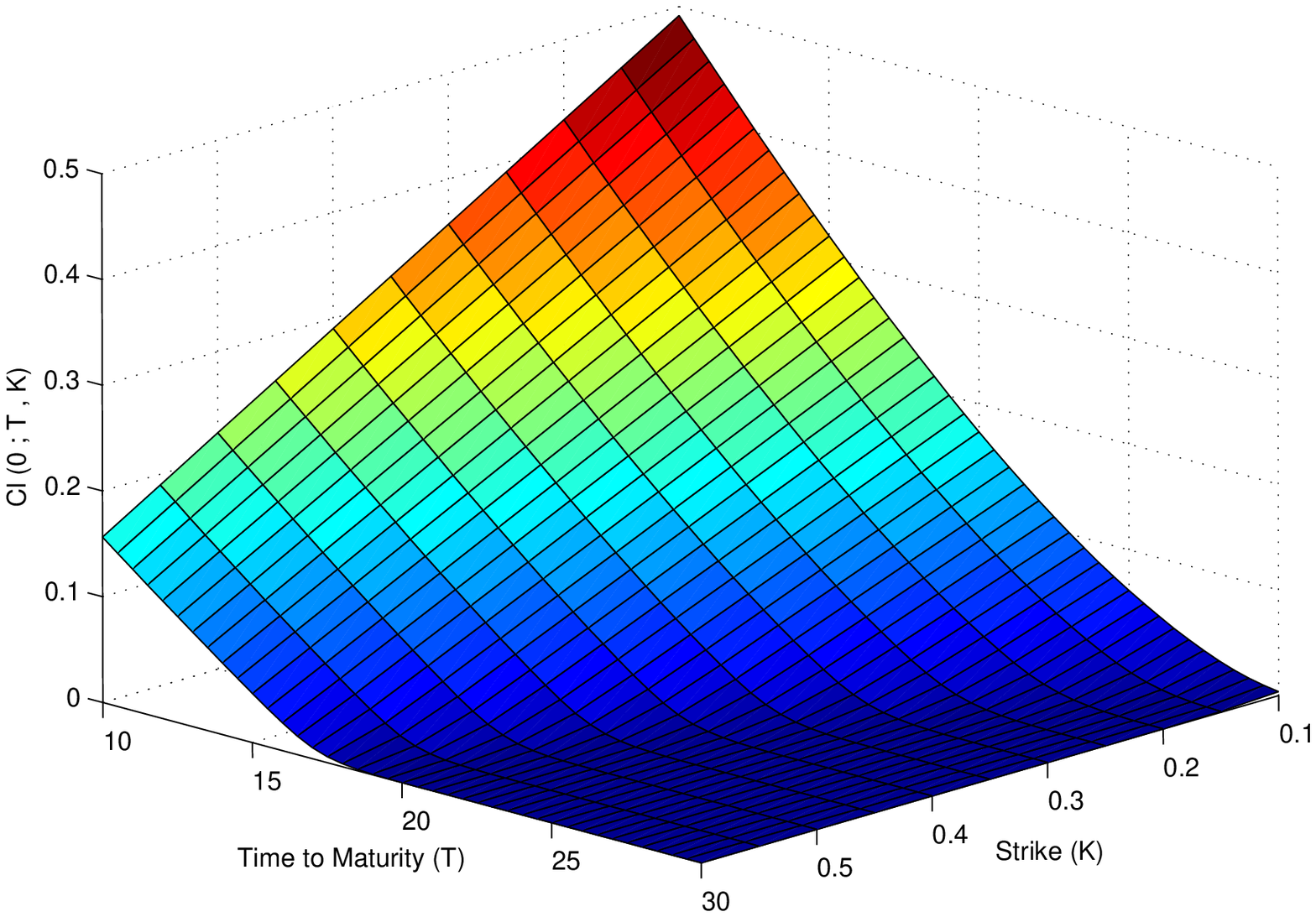}\includegraphics[width=8cm, height=6cm]{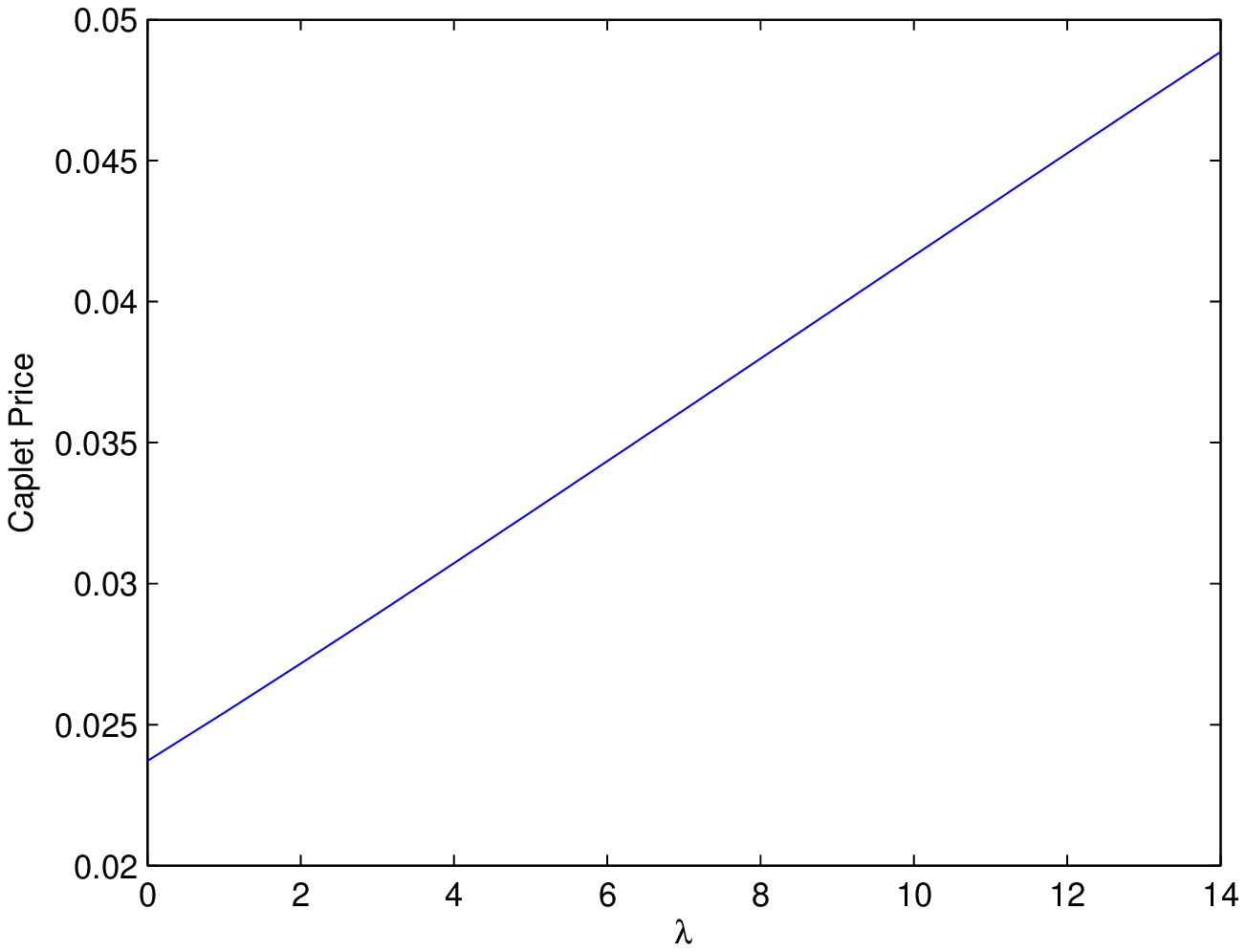}
  \caption{Caplet price as a function of (left panel) $T$ and $K$ and (right panel) $\lambda$ where $K=0.4$ and $T=20$.}\label{fig:caplet_price_plot}
    \end{center}
\end{figure}

The left panel of Figure~\ref{fig:caplet_price_plot} shows caplet
prices for a cohort aged $x=65$, using parameters as specified in
Table~\ref{table:modelparas}, as a function of time to maturity $T$
and strike $K$. We set $r=0.04$, $\lambda=8.5$ and $t=0$ such that
$\bar{S}_x(0,0) = 1$. A lower strike price indicates that the buyer
of a caplet is willing to pay more to secure a better protection
against a larger-than-expected survival probability. On the other
hand, when the time to maturity $T$ is increasing, the underlying
survival probability is likely to take smaller values, which leads to a higher
probability for the caplet to become out-of-the-money at maturity
for a fixed $K$, see Eq.~\eqref{eq:cappayoff}. Consequently, for a
fixed $K$ the caplet price decreases with increasing
$T$.

The right panel of Figure~\ref{fig:caplet_price_plot} illustrates
the effect of the market price of longevity risk $\lambda$ on the
caplet price. The price of a caplet increases with increasing $\lambda$.
As shown in Figure~\ref{fig:RiskAdjusted}, a larger value of
$\lambda$ will lead to an improvement in survival probability under
$\mathbb{Q}$. Thus, a higher caplet price is observed since the
underlying survival probability is larger (on average) under
$\mathbb{Q}$ when $\lambda$ increases, see
Eq.~\eqref{eq:caplet_expectedvalue}.

Since longevity cap is constructed as a portfolio of longevity
caplets, it can be priced as a sum of individual caplet prices,
see also Section~\ref{subsubsec:caphedged_portfolio}.

%% file: NumericalExamples.tex
\section{Managing Longevity Risk in a Hypothetical Life Annuity Portfolio}
\label{sec:hypotheticalexample}

Hedging features of a longevity swap and cap are examined for a
hypothetical life annuity portfolio subject to longevity risk.
Factors considered include the market price of longevity risk, the term
to maturity of hedging instruments and the size of the
underlying annuity portfolio.

\subsection{Setup}\label{sec:setup}

We consider a hypothetical life annuity portfolio that consists of a
cohort aged $x=65$. The size of the portfolio that corresponds to the 
number of policyholders, is denoted by $n$.  The underlying mortality intensity for the
cohort follows the two-factor Gaussian mortality model described in
Section~\ref{sec:model}, and the model parameters are specified in
Table~\ref{table:modelparas}. We assume that there is no loading for
the annuity policy and expenses are not included.

Further, we assume a single premium, whole life annuity of $\$1$ per year
payable in arrears conditional on the survival of the annuitant to
the payment dates. The fair value, or the premium, of the annuity
evaluated at $t=0$ is given by
\begin{equation}\label{eq:annuityprice}
a_{x} = \sum^{\omega-x}_{T=1} B(0,T) \, \tilde{S}_{x}(0,T)
\end{equation}
where $r=4\%$ and $\omega=110$ is the maximum age allowed in the
mortality model. The life annuity provider, thus, receives a total
premium, denoted by $A$, for the whole portfolio corresponding to
the sum of individual premiums:
\begin{equation}
A = n\,a_x.
\end{equation}
This is the present value of the asset held by the annuity provider
at $t=0$. Since the promised annuity cashflows depend on the death
times of annuitants in the portfolio, the present value of the
liability is subject to randomness caused by the stochastic dynamics
of the mortality intensity. The present value of the liability for
each policyholder, denoted by $L_k$, is determined by the death time
$\tau_k$ of the policyholder, and is given by
\begin{equation}
L_k = \sum^{\lfloor \tau_k \rfloor}_{T=1} B(0,T)
\end{equation}
for a simulated $\tau_k$, with $\lfloor q \rfloor$ denoting the next
smaller integer of a real number $q$. The present value of the
liability $L$ for the whole portfolio is obtained as a sum of
individual liabilities:
\begin{equation}
L = \sum^{n}_{k=1}L_k.
\end{equation}
The algorithm for simulating death times of annuitants, which
requires a single simulated path for the mortality intensity of the
cohort, is summarised in Appendix~\ref{A2}. The discounted surplus
distribution ($D_{\text{no}}$) of an unhedged annuity portfolio is
obtained by setting
\begin{equation}\label{disc_S_sample}
D_{\text{no}} = A - L.
\end{equation}
The impact of longevity risk is captured by simulating the
discounted surplus distribution where each sample is determined by
the realised mortality intensity of a cohort. Since traditional
pricing and risk management of life annuity relies on 
diversification effect, or the law of large numbers, we consider the discounted surplus distribution per policy
\begin{equation}
D_{\text{no}}/n.
\end{equation}
Figure~\ref{fig:lawlargeno} shows the discounted surplus
distribution per policy without longevity risk (i.e. when setting $\sigma_1
= \sigma =0$) with different portfolio sizes, varying from $n=2000$
to $8000$. As expected, the mean of the distribution is centred
around zero as there is no loading assumed in the pricing algorithm,
while the standard deviation diminishes as the number of policies
increases.

In the following we consider a longevity swap and a cap as hedging
instruments. These are index-based instruments where the payoffs
depend on the survivor index, or the realised survival probability
(Eq.~\eqref{eq:survivorindex}), which is in turn
determined by the realised mortality intensity. We do not consider basis risk\footnote{If basis risk is present, we need to distinguish between the mortality intensity for the population ($\mu_x^I$) and mortality intensity for the cohort ($\mu_x$) underlying the annuity portfolio, see \cite{BiffisBlPiAr14}.} but due to a finite portfolio size, the actual proportion of survivors, $\frac{n-N_t}{n}$, where $N_t$ denotes the number of deaths experienced by a cohort during the period $[0,t]$, will be in general similar, but not identical, to the survivor index (Appendix~\ref{A2}). As a result, the static hedge will be able to reduce systematic mortality risk, whereas the idiosyncratic mortality risk component will be retained by the annuity provider.

\begin{figure}[h]
  \begin{center}
   \includegraphics[width=12cm, height=9cm]{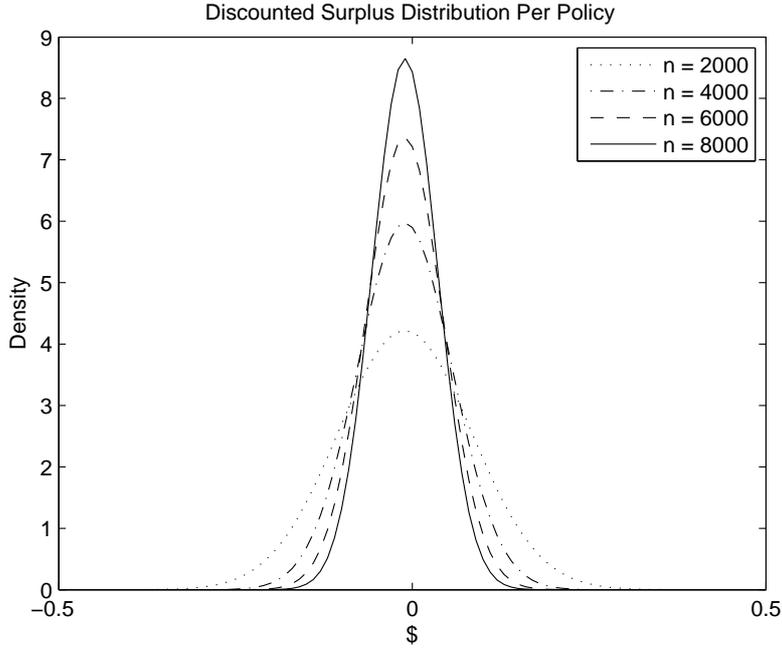}
  \caption{Discounted surplus distribution per policy without longevity risk with different portfolio size ($n$).}\label{fig:lawlargeno}
    \end{center}
\end{figure}

\subsubsection{A Swap-Hedged Annuity Portfolio}
For an annuity portfolio hedged by an index-based longevity swap, payments from the swap
\begin{equation}\label{eq:swappayoff}
n\, \left( e^{-\int^{T}_0\mu_x(v)\,dv}-K(T) \right)
\end{equation}
at time $T \in\{1,...,\hat{T}\}$ depend on the realised mortality
intensity, where $\hat{T}$ denotes the term to maturity of the longevity
swap. The number of policyholders $n$ acts as the notional amount of
the swap contract so that the quantity $n
\exp\{-\int^{T}_0\mu_x(v)\,dv\}$ represents the number of survivors
implied by the realised mortality intensity at time $T$. We fix the
strike of a swap to the risk-adjusted survival probability, that is,
\begin{equation}
K(T)= \tilde{S}_x(0,T)=E^\mathbb{Q}_0\left(e^{-\int^{T}_0\mu_x(v)\,dv}\right)
\end{equation}
such that the price of a swap is zero at $t=0$, see Section~\ref{subsec:swaps}.
The discounted surplus distribution of a swap-hedged annuity portfolio can be expressed as
\begin{equation}
D_{\text{swap}} = A - L + F_{\text{swap}}
\end{equation}
where
\begin{equation}\label{eq:Fswap}
F_{\text{swap}} = n\,\sum^{\hat{T}}_{T=1} B(0,T)\,  \left( e^{-\int^{T}_0\mu_x(v)\,dv}-\tilde{S}_x(0,T) \right)
\end{equation}
is the (random) discounted cashflow coming from a long position in
the longevity swap. The discounted surplus distribution per policy
of a swap-hedged annuity portfolio is determined by
$D_{\text{swap}}/n$.

\subsubsection{A Cap-Hedged Annuity Portfolio}\label{subsubsec:caphedged_portfolio}
For an annuity portfolio hedged by an index-based longevity cap, the cashflows
\begin{equation}
n\,\max{\left\{\left(e^{-\int^{T}_0\mu_x(v)\,dv}-K(T)\right),0\right\}}
\end{equation}
at $T \in\{1,...,\hat{T}\}$ are payments from a long position in the longevity cap. We set
\begin{equation}
K(T) = S_x(0,T)=E^\mathbb{P}_0\left(e^{-\int^{T}_0\mu_x(v)\,dv}\right)
\end{equation}
such that the strike for a longevity caplet is the ``best estimated"
survival probability given $\mathcal{F}_0$.\footnote{For a longevity
swap, the risk-adjusted survival probability is used as a strike price so
that the price of a longevity swap is zero at inception. In
contrast, a longevity cap has non-zero price and $S_x(0,T)$ is the
most natural choice for a strike.}  The discounted surplus
distribution of a cap-hedged annuity portfolio is given by
\begin{equation}
D_{\text{cap}} = A - L + F_{\text{cap}} - C_{\text{cap}}
\end{equation}
where
\begin{equation}\label{eq:Fcap}
F_{\text{cap}} = n\,\sum^{\hat{T}}_{T=1} B(0,T)\, \max{\left\{\left(e^{-\int^{T}_0\mu_x(v)\,dv}-S_x(0,T)\right),0\right\}}
\end{equation}
is the (random) discounted cashflow from holding the longevity cap
and
\begin{equation}
C_{\text{cap}} = n\, \sum^{\hat{T}}_{T=1}C\ell\left(0; T,S_x(0,T)\right)
\end{equation}
is the price of the longevity cap. The discounted surplus distribution per policy of a cap-hedged annuity portfolio is given by $D_{\text{cap}}/n$.

\subsection{Results}

Hedging results are summarised by means of summary statistics that
include mean, standard deviation (std. dev.), skewness, as well as Value-at-Risk (VaR) and Expected Shortfall (ES) of the discounted surplus
distribution per policy of an unhedged, a swap-hedged and
a cap-hedged annuity portfolio. Skewness is included since the
payoff of a longevity cap is nonlinear and the resulting
distribution of a cap-hedged annuity portfolio is not symmetric. VaR
is defined as the $q$-quantile of the discounted surplus
distribution per policy. ES is defined as the expected loss of the
discounted surplus distribution per policy given the loss is at or
below the $q$-quantile. We fix $q = 0.01$ so that the confidence
interval for VaR and ES corresponds to $99\%$. We use 5,000
simulations to obtain the distribution for the discounted surplus.
Hedge effectiveness is examined with respect to (w.r.t.) different
assumptions underlying the market price of longevity risk
($\lambda$), the term to maturity of hedging instruments ($\hat{T}$)
and the portfolio size ($n$). Parameters for the base case are as
specified in Table~\ref{table:basecase}.

\begin{table}[!ht]
\center \setlength{\tabcolsep}{1em}
\renewcommand{\arraystretch}{1.1}
\small{\caption{Parameters for the base case.}\label{table:basecase}
\begin{tabular}{lll}
\hline
\hline
$\lambda$  & $\hat{T}$ (years) & $n$  \\
\hline
8.5 & { }{ }{ }{ }30   & 4000      \\
\hline
\hline
\end{tabular}}
\end{table}

\subsubsection{Hedging Features w.r.t. Market Price of Longevity Risk}
\label{subsubsec:hedge_pricerisk}

\begin{figure}[!ht]
  \begin{center}
   \includegraphics[width=15.5cm, height=12cm]{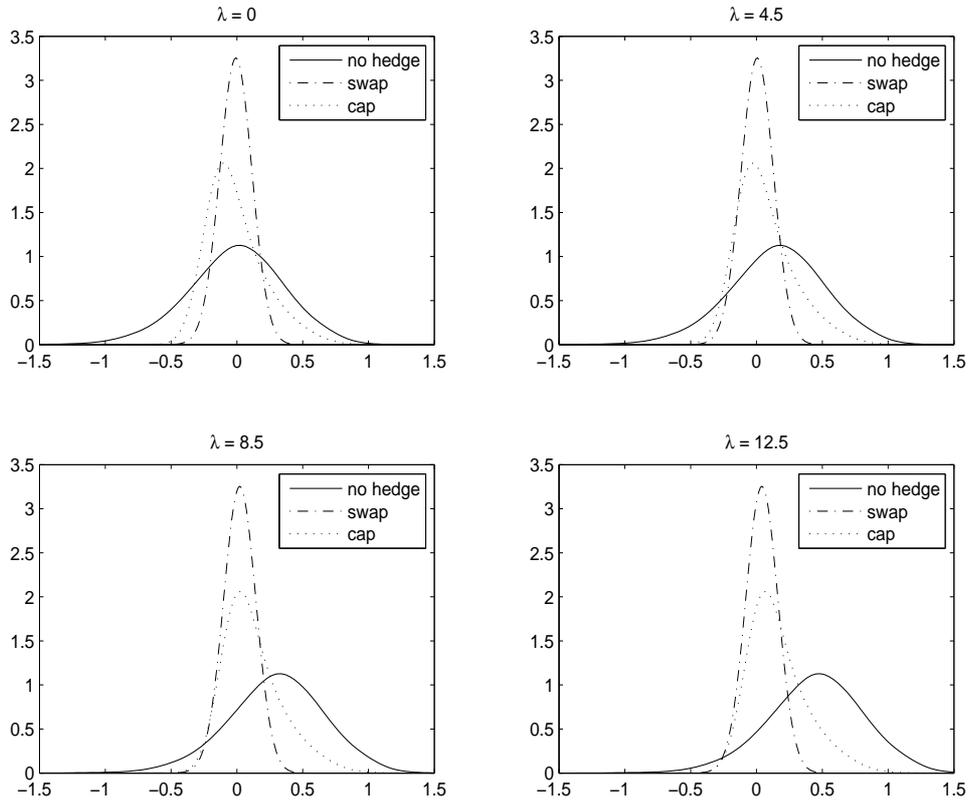}
  \caption{Effect of the market price of longevity risk $\lambda$ on the discounted surplus distribution per policy.}\label{fig:hedge_pricerisk}
    \end{center}
\end{figure}

\begin{table}[h]
\center\small{\caption{\label{tab:hedge_pricerisk}Hedging features of a longevity swap and cap w.r.t. market price of longevity risk $\lambda$.}
\begin{tabular*}{1.0\textwidth}%
     {@{\extracolsep{\fill}}l|rrrrr}
\hline \hline
  &  Mean & Std.dev. & Skewness &VaR$_{0.99}$ & ES$_{0.99}$ \\
\hline
&\multicolumn{5}{c}{$\lambda = 0$}\\
\hline
No hedge & -0.0076 &  0.3592 & -0.2804 & -0.9202 & -1.1027 \\
Swap-hedged & -0.0089 &  0.0718 & -0.1919 & -0.1840 & -0.2231 \\
Cap-hedged & -0.0086 &  0.2054 & 1.0855 &  -0.3193 & -0.3515 \\
\hline
&\multicolumn{5}{c}{$\lambda = 4.5$}\\
\hline
No hedge & 0.1520 &  0.3592 & -0.2804 & -0.7606 & -0.9431 \\
Swap-hedged & 0.0048 &  0.0718 & -0.1919 & -0.1703 & -0.2094 \\
Cap-hedged & 0.0682 &  0.2054 & 1.0855 &  -0.2425 & -0.2746 \\
\hline
&\multicolumn{5}{c}{$\lambda = 8.5$}\\
\hline
No hedge & 0.2978 &  0.3592 & -0.2804 & -0.6148 & -0.7973 \\
Swap-hedged & 0.0204 &  0.0718 & -0.1919 & -0.1547 & -0.1938 \\
Cap-hedged & 0.1205 &  0.2054 & 1.0855 &  -0.1903 & -0.2224 \\
\hline
&\multicolumn{5}{c}{$\lambda = 12.5$}\\
\hline
No hedge & 0.4475 &  0.3592 & -0.2804 & -0.4650 & -0.6476 \\
Swap-hedged & 0.0398 &  0.0718 & -0.1919 & -0.1354 & -0.1744 \\
Cap-hedged & 0.1619 &  0.2054 & 1.0855 &  -0.1489 & -0.1810 \\
\hline \hline
\end{tabular*} }
\end{table}

The market price of longevity risk $\lambda$ is one of the factors
that determines prices of longevity derivatives and life annuity
policies. Since payoffs of a longevity swap, a cap and a life annuity
are contingent on the same underlying mortality intensity of a
cohort, all these products are priced using the same $\lambda$.
Figure~\ref{fig:hedge_pricerisk} and Table~\ref{tab:hedge_pricerisk}
illustrate the effect of changing $\lambda$ on the distributions of an
unhedged, a swap-hedged and a cap-hedged annuity portfolio. The
degree of longevity risk can be quantified by the standard
deviation, the VaR and the ES of the distributions. We observe that
increasing $\lambda$ leads to the shift of the distribution to the
right, resulting in a higher average surplus. On the other hand,
changing $\lambda$ has no impact on the standard deviation and
the skewness of the distribution.

For an unhedged annuity portfolio, a higher $\lambda$ leads to
higher premium for the life annuity policy since the annuity price
is determined by the risk-adjusted survival probability
$\tilde{S}_x(0,T)$, see Eq.~\eqref{eq:annuityprice}. In other words,
an increase in the annuity price compensates the provider for the
longevity risk undertaken when selling life annuity policies. There
is also a trade-off between risk premium and affordability. Setting
a higher premium will clearly improve the risk and return of an
annuity business, it might, however, reduce the interest of potential
policyholders. An empirical relationship between implied longevity
and annuity prices is studied in \cite{ChMiSa14}.

When life annuity portfolio is hedged using a longevity swap, the standard deviation and the absolute values of the VaR and the ES reduce substantially. The higher return obtained by charging a larger market price of longevity risk in life annuity policies is offset by an increased price paid implicitly in the swap contract (since $\tilde{S}_x(0,T) \geq S_x(0,T)$ in Eq.~\eqref{eq:Fswap}). It turns out that as $\lambda$ increases an extra return earned in the annuity portfolio and the higher implicit cost of the longevity swap nearly offset each other out on average. The net effect is that a swap-hedged annuity portfolio remains to a great extent unaffected by the assumption on $\lambda$, leading only to a very minor increase in the mean of the distribution.

For a cap-hedged annuity portfolio, the discounted surplus distribution is positively
skewed since a longevity cap allows an annuity provider to get
exposure to the upside potential when policyholders live shorter
than expected. Compared to an unhedged portfolio, the standard
deviation and the absolute values of the VaR and the ES are also
reduced but the reduction is smaller compared to a swap-hedged
portfolio. When $\lambda$ increases, we observe that the mean of the
distribution for a cap-hedged portfolio increases faster than for a
swap-hedged portfolio but slower than for an unhedged portfolio. It
can be explained by noticing that when the survival probability of a
cohort is overestimated, that is, when annuitants turn out to live
shorter than expected, holding a longevity cap has no effect
(besides paying the price of a cap for longevity protection at the
inception of the contract) while there is a cash outflow when
holding a longevity swap, see Eq.~\eqref{eq:Fswap} and
Eq.~\eqref{eq:Fcap}.

In the longevity risk literature, the VaR and the ES are of a
particular importance as they are the main factors determining the
capital reserve when dealing with exposure to longevity risk
(\cite{MeyrickeSh14}). As shown in Table~\ref{tab:hedge_pricerisk},
the difference between a swap-hedged and a cap-hedged portfolio in
terms of the VaR and the ES becomes smaller when $\lambda$
increases. In fact, for $\lambda \geq 17.5$, a longevity cap becomes
more effective in reducing the tail risk of an annuity portfolio
compared to a longevity swap.\footnote{Given $\lambda = 17.5$, the
VaR and the ES for a swap-hedged portfolio are $-0.1051$ and
$-0.1441$ respectively. For a cap-hedged portfolio they become
$-0.1038$ and $-0.1360$, respectively.} This result suggests that a
longevity cap, besides being able to capture the upside potential,
can be a more effective hedging instrument than a longevity swap in
terms of reducing the VaR and the ES when the demanded market price
of longevity risk $\lambda$ is large.

\subsubsection{Hedging Features w.r.t. Term to Maturity}

\begin{figure}[!ht]
  \begin{center}
   \includegraphics[width=15.5cm, height=12cm]{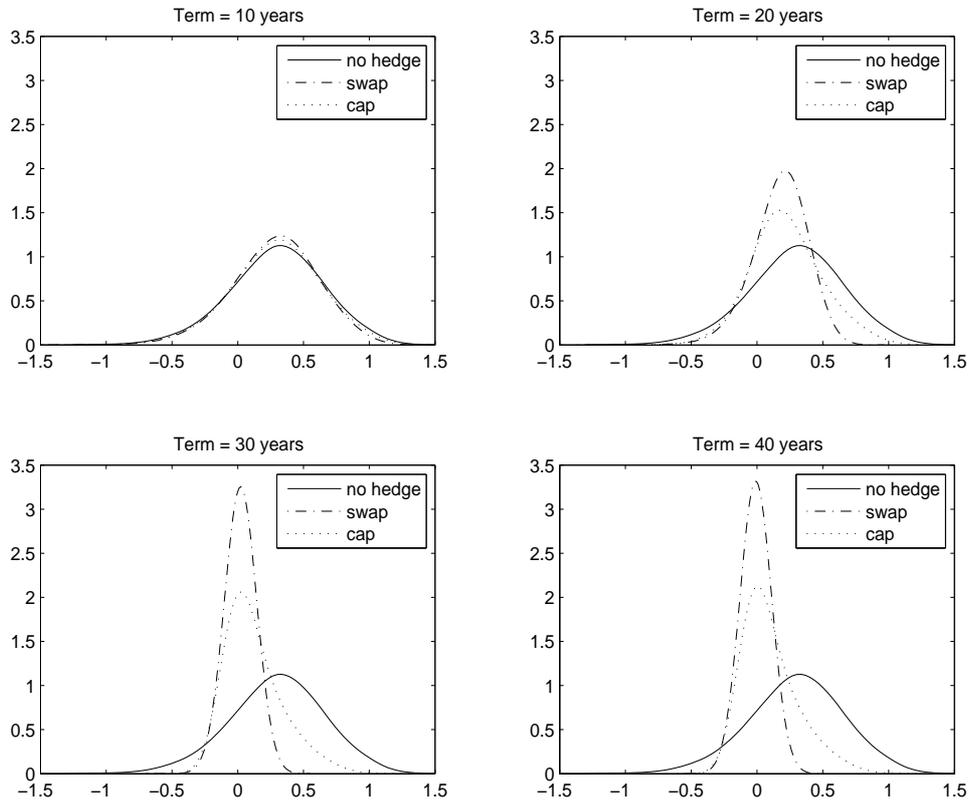}
  \caption{Effect of the term to maturity $\hat{T}$ of the hedging instruments on the discounted surplus distribution per policy.}\label{fig:hedge_term}
    \end{center}
\end{figure}

\begin{table}[h]
\center\small{\caption{\label{tab:hedge_term} Hedging features of a longevity swap and cap w.r.t. term to maturity $\hat{T}$.}
\begin{tabular*}{1.0\textwidth}%
     {@{\extracolsep{\fill}}l|rrrrr}
\hline \hline
  &  Mean & Std.dev. & Skewness &VaR$_{0.99}$ & ES$_{0.99}$ \\
\hline
&\multicolumn{5}{c}{$\hat{T} = 10$ Years}\\
\hline
No hedge & 0.2978 &  0.3592 & -0.2804 & -0.6148 & -0.7973 \\
Swap-hedged & 0.2820 &  0.2911 & -0.3871 & -0.5707 & -0.7490 \\
Cap-hedged & 0.2893 &  0.2989 & -0.2661 &  -0.5801 & -0.7592 \\
\hline
&\multicolumn{5}{c}{$\hat{T} = 20$ Years}\\
\hline
No hedge & 0.2978 &  0.3592 & -0.2804 & -0.6148 & -0.7973 \\
Swap-hedged & 0.1740 &  0.1794 & -0.7507 & -0.3656 & -0.5061 \\
Cap-hedged & 0.2234 &  0.2310 & 0.2006 &  -0.3870 & -0.5259 \\
\hline
&\multicolumn{5}{c}{$\hat{T} = 30$ Years}\\
\hline
No hedge & 0.2978 &  0.3592 & -0.2804 & -0.6148 & -0.7973 \\
Swap-hedged & 0.0204 &  0.0718 & -0.1919 & -0.1547 & -0.1938 \\
Cap-hedged & 0.1205 &  0.2054 & 1.0855 &  -0.1903 & -0.2224 \\
\hline
&\multicolumn{5}{c}{$\hat{T} = 40$ Years}\\
\hline
No hedge & 0.2978 &  0.3592 & -0.2804 & -0.6148 & -0.7973 \\
Swap-hedged & -0.0091 &  0.0668 & 0.0277 & -0.1616 & -0.1869 \\
Cap-hedged & 0.0984 &  0.1999 & 1.1527 &  -0.1909 & -0.2131 \\
\hline \hline
\end{tabular*} }
\end{table}

Table~\ref{tab:hedge_term} and Figure~\ref{fig:hedge_term} summarize
hedging results with respect to the term to maturity of hedging
instruments. Due to the long-term nature of the contracts, the hedges
are ineffective for $\hat{T} \leq 10$ years and the standard
deviations are reduced only by around $17-19\%$ for both
instruments. The lower left panel of Figure~\ref{fig:sim_mu} shows
that there is little randomness around the realised survival
probability for the first few years for a cohort aged 65, and
consequently the hedges are insignificant when $\hat{T}$ is short.

The difference in hedge effectiveness between $\hat{T}= 30$ and
$\hat{T} =40$ for both instruments is also insignificant. In fact,
the longevity risk underlying the annuity portfolio becomes small
after $30$ years since the majority of annuitants has already
deceased before reaching the age of 95. In our model setup the
chance for a 65 years old to live up to 95 is around $6\%$
(Figure~\ref{fig:RiskAdjusted} with $\lambda=0$) and, hence, only
around $4000 \times 6\% = 240$ policies will still be in-force after
30 years. Much of the risk left is attributed to idiosyncratic
mortality risk, and hedging longevity risk for a small portfolio
using index-based instruments is of limited use.

For a swap-hedged portfolio, the standard deviation is reduced significantly when $\hat{T} > 20$ years. The mean surplus, on the other hand, drops to nearly zero since there is a higher cost implied for the hedge with increasing number of S-forwards involved to form the swap as $\hat{T}$ increases.

Similar hedging features with respect to $\hat{T}$ are observed for
a longevity cap. However, the skewness of the distribution of a
cap-hedged portfolio increases with increasing $\hat{T}$. It can be
explained by noticing that while a longevity cap is able to capture
the upside potential regardless of $\hat{T}$, it provides a better
longevity risk protection when $\hat{T}$ is larger. As a result, the
distribution of a cap-hedged portfolio becomes more asymmetric when
$\hat{T}$ increases.

\subsubsection{Hedging Features w.r.t. Portfolio Size}\label{subsubsec:portfoliosize}

\begin{figure}[!ht]
  \begin{center}
   \includegraphics[width=15.5cm, height=12cm]{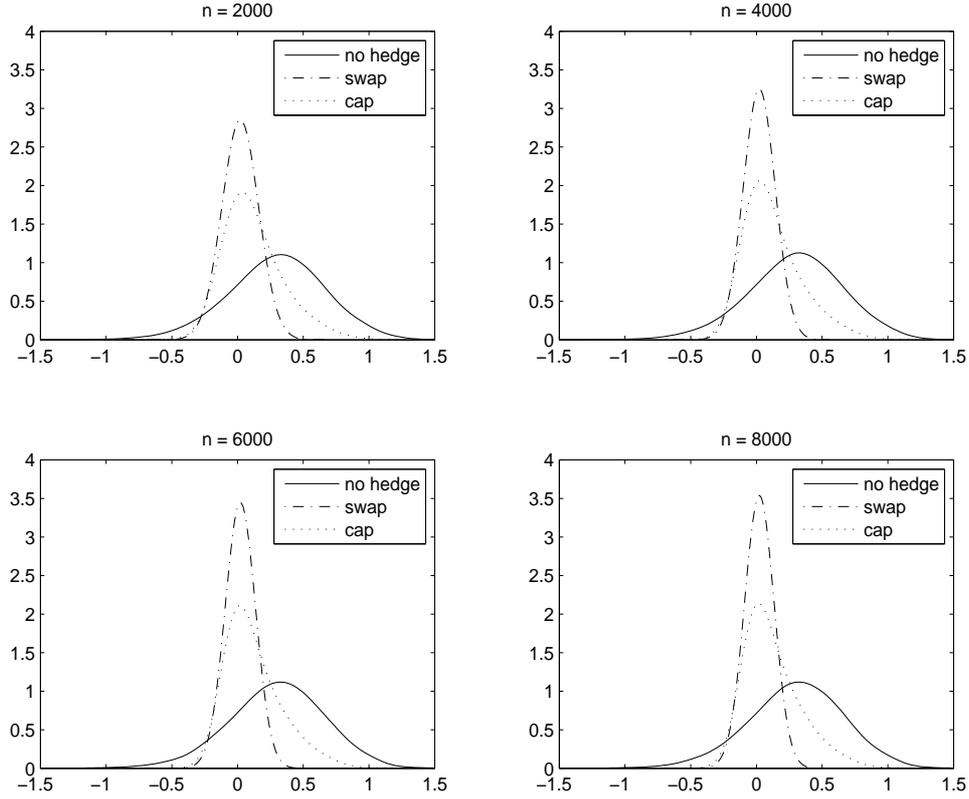}
  \caption{Effect of the portfolio size $n$ on the discounted surplus distribution per policy.}\label{fig:portfoliosize}
    \end{center}
\end{figure}

\begin{table}[h]
\center\small{\caption{\label{tab:hedge_portfoliosize} Hedging features of a longevity swap and cap w.r.t. different portfolio size ($n$).}
\begin{tabular*}{1.0\textwidth}%
     {@{\extracolsep{\fill}}l|rrrrr}
\hline \hline
  &  Mean & Std.dev. & Skewness &VaR$_{0.99}$ & ES$_{0.99}$ \\
\hline
&\multicolumn{5}{c}{$n = 2000$}\\
\hline
No hedge & 0.2973 &  0.3646 & -0.2662 & -0.6360 & -0.8107 \\
Swap-hedged & 0.0200 &  0.0990 & -0.1615 & -0.2120 & -0.2653 \\
Cap-hedged & 0.1200 &  0.2160 & 0.9220 &  -0.2432 & -0.2944 \\
\hline
&\multicolumn{5}{c}{$n = 4000$}\\
\hline
No hedge & 0.2978 &  0.3592 & -0.2804 & -0.6148 & -0.7973 \\
Swap-hedged & 0.0204 &  0.0718 & -0.1919 & -0.1547 & -0.1938 \\
Cap-hedged & 0.1205 &  0.2054 & 1.0855 &  -0.1903 & -0.2224 \\
\hline
&\multicolumn{5}{c}{$n = 6000$}\\
\hline
No hedge & 0.2977 &  0.3566 & -0.2786 & -0.6363 & -0.8001 \\
Swap-hedged & 0.0204 &  0.0594 & -0.3346 & -0.1259 & -0.1660 \\
Cap-hedged & 0.1204 &  0.2016 & 1.1519 &  -0.1639 & -0.2051 \\
\hline
&\multicolumn{5}{c}{$n = 8000$}\\
\hline
No hedge & 0.2982 &  0.3554 & -0.2920 & -0.6060 & -0.7876 \\
Swap-hedged & 0.0209 &  0.0536 & -0.5056 & -0.1190 & -0.1595 \\
Cap-hedged & 0.1209 &  0.1992 & 1.1616 &  -0.1598 & -0.1991 \\
\hline \hline
\end{tabular*} }
\end{table}

Table~\ref{tab:hedge_portfoliosize} and
Figure~\ref{fig:portfoliosize} demonstrate hedging features of a
longevity swap and a cap with changing portfolio size $n$. We
observe a decrease in standard deviation, as well as the VaR and the
ES (in absolute terms) when portfolio size increases. Compared to an
unhedged portfolio, the reduction in the standard deviation and the
risk measures is larger for a swap-hedged portfolio, compared to a
cap-hedged portfolio. Recall that idiosyncratic mortality risk
becomes significant when $n$ is small. We quantify the effect of the
portfolio size on hedge effectiveness by introducing the 
measure of longevity risk reduction $R$, defined in terms of the variance
of the discounted surplus per policy, that is,
\begin{equation}
R = 1- \frac{\text{Var}(\bar{D}^*)}{\text{Var}(\bar{D})},
\end{equation}
where $\text{Var}(\bar{D}^*)$ and $\text{Var}(\bar{D})$ represent
the variances of the discounted surplus distribution per policy for
a hedged and an unhedged annuity portfolio, respectively. The results
are reported in Table~\ref{tab:R}.

\begin{table}[h]
\center \setlength{\tabcolsep}{1em}
\renewcommand{\arraystretch}{1.1}
\center\small{\caption{\label{tab:R} Longevity risk reduction $R$ of a longevity swap and cap w.r.t. different portfolio size ($n$).}}
\begin{tabular}{lllll}
\hline \hline
$n$  &  2000 & 4000 & 6000 & 8000 \\
\hline
$R_\text{swap}$ & $92.6\%$ & $96.0\%$ & $97.2\%$ & $97.7\%$  \\
\hline
$R_\text{cap}$ & $64.9\%$ & $67.3\%$ & $68.0\%$ & $68.6\%$ \\
\hline \hline
\end{tabular}
\end{table}

\cite{LiHa11} consider hedging longevity risk using a portfolio of
q-forwards and find the longevity risk reduction of $77.6\%$ and
$69.6\%$ for portfolio size of 10,000 and 3,000, respectively. In
contrast to \cite{LiHa11}, we do not consider basis risk and the
result of using longevity swap as a hedging instrument leads to a
greater risk reduction. Overall, our results indicate that hedge
effectiveness for an index-based longevity swap and a cap diminishes
with decreasing $n$ since idiosyncratic mortality risk cannot be
effectively diversified away for a small portfolio size. Even though
a longevity cap is less effective in reducing the variance, part of
the dispersion is attributed to its ability of capturing the upside
of the distribution when survival probability of a cohort is
overestimated. From Table~\ref{tab:hedge_portfoliosize} we also 
observe that the distribution becomes more positively skewed for a cap-hedged portfolio when $n$ increases, which is
a consequence of having a larger exposure to longevity risk with
increasing number of policyholders in the portfolio.

%% file: Conclusion.tex
\section{Conclusion}

Life and pension annuities are the most important types of
post-retirement products offered by annuity providers to help
securing lifelong incomes for the rising number of retirees. While
interest rate risk can be managed effectively in the financial
markets, longevity risk is a major concern for annuity providers as
there are only limited choices available to mitigate the long-term risk.
Development of effective financial instruments for longevity risk in
capital markets is arguably the best solution available.

Two types of longevity derivatives, a longevity swap and a cap, are
analysed in this paper from a pricing and hedging perspective. We
apply a tractable Gaussian mortality model to capture the longevity
risk, and derive explicit formulas for important quantities such as
survival probabilities and prices of longevity derivatives. Hedge
effectiveness and features of an index-based longevity swap and a
cap used as hedging instruments are examined using a hypothetical
life annuity portfolio exposed to longevity risk.

Our results suggest that the market price of longevity risk
$\lambda$ is a small contributor to hedge effectiveness of a
longevity swap since a higher annuity price is partially
offset by an increased cost of hedging when $\lambda$ is taken into
account. It is shown that a longevity cap, while being able to
capture the upside potential when survival probabilities are
overestimated, can be more effective in reducing longevity tail risk
compared to a longevity swap, provided that $\lambda$ is large
enough. The term to maturity $\hat{T}$ is an important factor in
determining hedge effectiveness. However, the difference in hedge
effectiveness is only marginal when $\hat{T}$ increases from $30$ to
$40$ years for an annuity portfolio consisting of a single cohort
aged 65 initially. This is due to the fact that only a small number
of policies will still be in-force after a long period of time (30
to 40 years), and index-based instruments turn out to be ineffective
when idiosyncratic mortality risk becomes a larger contributor to
the overall risk, compared to systematic mortality risk. The effect of
the portfolio size $n$ on hedge effectiveness is quantified and
compared with the result obtained in \cite{LiHa11} where population
basis risk is taken into account. In addition, we find that the
skewness of the surplus distribution of a cap-hedged portfolio is
sensitive to the term to maturity and the portfolio size, and, as a
result, the difference between a longevity swap and a cap when used
as hedging instruments becomes more pronounced for larger $\hat{T}$
and $n$.

As discussed in \cite{BiffisBl14}, developing a liquid longevity
market requires reliable and well-designed financial instruments
that can attract sufficient amount of interests from both buyers and
sellers. Besides of a longevity swap, which is so far a common
longevity hedging choice for annuity providers, option-type
instruments such as longevity caps can provide hedging features that
linear products cannot offer. A longevity cap is shown to have
alternative hedging properties compared to a swap, and this
option-type instrument would also appeal to certain classes of
investors interested in receiving premiums by selling a longevity
insurance. Further research on the design of longevity-linked
instruments from the perspectives of buyers and sellers would
provide a further step towards the development of an active longevity
market.

%% file: AppendixA2.tex
\section{Appendix} \label{A2}
To simulate death times of annuitants, we notice that once a sample of the mortality intensity is obtained, the Cox process becomes an inhomogeneous Poisson process and the first jump times, which are interpreted as death times, can be simulated as follows (see e.g. \cite{BrigoMe07}):
\begin{enumerate}
  \item Simulate the mortality intensity $\mu_x(t)$ from $t=0$ to $t=\omega-x$.
  \item Generate a standard exponential random variable $\xi$. For example, using an inverse transform method, we have $\xi = -\ln{(1-u)}$ where $u \sim \text{Uniform}(0,1)$.
  \item Set the death time $\tau$ to be the smallest $T$ such that $\xi \leq \int^T_0 \mu_x(s) \,ds$. If $\xi > \int^{\omega-x}_0 \mu_x(s) \,ds$ then set $\tau = \omega-x$. 
  \item Repeat step (2) and (3) to obtain another death time.
\end{enumerate}
The payoff of an index-based hedging instrument depends on the
realised survival probability $\exp\{-\int^t_0 \mu_x(v)\,dv\}$. The
payoff of a customised instrument, on the other hand, depends
on the proportion of survivors, $\frac{n-N_t}{n}$, underlying an
annuity portfolio where the number of deaths, $N_t$, is obtained by
counting the number of simulated death times that are smaller than
$t$. Note that
\begin{equation}
e^{-\int^t_0 \mu_x(v)\,dv} \approx \frac{n-N_t}{n}
\end{equation}
and the accuracy of the approximation improves when $n$ increases.
